\documentclass{llncs} %
\title{Quantifying Attacker Capability Via Model Checking Multiple Properties\\(Extended Version)}
\author{Eric Rothstein-Morris \and Sun Jun}
\institute{Singapore University of Technology and Design }
\usepackage{ifthen}
\usepackage{fixltx2e}
\usepackage{pifont}
\DeclareMathAlphabet{\mathpzc}{OT1}{pzc}{m}{it}
\usepackage{amssymb} 
\usepackage{todonotes} 
\usepackage{multirow}
\usepackage{mathrsfs}
\usepackage[]{algorithm2e}
\LinesNumbered
\SetKwRepeat{Do}{do}{while}
\SetKw{Break}{break}

\usepackage{diagrams}
\usepackage{color}
\usepackage{graphicx}
\usepackage{color}
\usepackage{verbatim}
\usepackage{graphicx}
\usepackage{eurosym}
\usepackage{bm}
\usepackage{array}
\usepackage{float}
\usepackage{ stmaryrd }
\usepackage{braket}
\usepackage{tikz}
\usepackage{tikz-cd}
\usepackage{mathtools}
\usepackage{listings}
\usepackage{mdframed}

\usepackage{tabularx}

\usepackage{semantic}

\newcommand{\vect}[1]{\ensuremath{\overrightarrow{\bm{#1}}}}
\newcommand{\Verify}{{\AsFunction{Verify}}}

\newcommand{\addIt}{{\AsFunction{add}}}%{\bigcirc}%

\newcommand{\crtu}{{\AsFunction{ctu}}}

\newcommand{\TheCounterexample}{{c_{ex}}}

\newcommand{\shift}{{\AsFunction{shift}}}%{{\AsFunction{shift}^{\rightarrow}}}
\newcommand{\id}{{\AsFunction{id}}}

\newcommand{\AsRule}[1]{{{\textsc{#1}}}}
\newcommand{\AsFunction}[1]{{\mathtt{#1}}}

\newcommand{\True}{\mathtt{tt}}%{\mathtt{{True}}}

\newcommand{\False}{\mathtt{ff}}%{\mathtt{{False}}}

\newcommand{\Observation}{\downarrow\!}%{\AsFunction{{obs}}}

\newcommand{\Pack}{\AsFunction{{pack}}}

\MakeRobust{\Pack}
\newcommand{\Powerset}{{\mathscr{P}}}
\newcommand{\FinitePowerset}{{\Powerset_\omega}}

\newcommand{\TheSetOfVariables}{\AsFunctor{V}}

\newcommand{\TheVariable}{\textsc{v}}

\newcommand{\TheSetOfFormulae}{{\mathbb{L}}}
\newcommand{\TheSetOfSafetyFormulae}{{L}}

\newcommand{\TheSetOfClosedAndGuardedSafetyFormulae}{{\TheSetOfSafetyFormulae^{c}_{g}}}
\newcommand{\TheSetOfClosedAndGuardedFormulae}{{\TheSetOfFormulae^{c}_{g}}}

\newcommand{\SomeElement}[1]{
	{
	\ifthenelse{\equal{#1}{1}}
		{{x}}
		{\ifthenelse{\equal{#1}{2}}
		{{y}}
		{\ifthenelse{\equal{#1}{3}}
		{{z}}
		{{{#1}}}}}
	}
}
\MakeRobust{\SomeElement}
\newcommand{\TheElement}{{\SomeElement{1}}}

\newcommand{\AnotherElement}{{\SomeElement{2}}}

\newcommand{\SomeState}[1]{
	\mathtt{
	\ifthenelse{\equal{#1}{1}}
		{{s}}
		{\ifthenelse{\equal{#1}{2}}
		{{s'}}
		{\ifthenelse{\equal{#1}{3}}
		{{s''}}
		{{{#1}}}}}
	}
}
\MakeRobust{\SomeState}

\newcommand{\SomeSet}[1]{
	{\ifthenelse{\equal{#1}{1}}
		{X}%{\mathtt{X}}
		{\ifthenelse{\equal{#1}{2}}
		{Y}%{\mathtt{Y}}
		{\ifthenelse{\equal{#1}{3}}
		{Z}%{\mathtt{Z}}
		{{#1}}}}
	}
}
\MakeRobust{\SomeSet}

\newcommand{\TheSet}{{\SomeSet{1}}}
\newcommand{\OtherSet}{{\SomeSet{2}}}
\newcommand{\AnotherSet}{{\SomeSet{3}}}
\newcommand{\AsSet}[1]{{{\mathbf{#1}}}}%{{{#1_\TheCategoryOfSets}}}

\newcommand{\ThePowersetOf}[1]{{\Powerset\!\left({#1}\right)}}
\newcommand{\TheFinitePowersetOf}[1]{{\FinitePowerset\left(#1\right)}}

\newcommand{\SomePredicate}[1]{
	{\ifthenelse{\equal{#1}{1}}
		{{{P}}}
		{\ifthenelse{\equal{#1}{2}}
		{{{Q}}}
		{\ifthenelse{\equal{#1}{3}}
		{{{R}}}
		{{#1}}}}
	}
}
\MakeRobust{\SomePredicate}

\newcommand{\ThePredicate}{{\SomePredicate{1}}}
\newcommand{\OtherPredicate}{{\SomePredicate{2}}}

\newcommand{\TheGuardedFormula}{\psi}
\newcommand{\SomeFormula}[1]{
	{\ifthenelse{\equal{#1}{1}}
		{\Phi}
		{\ifthenelse{\equal{#1}{2}}
		{{\phi}}
		{\ifthenelse{\equal{#1}{3}}
		{{\AsRule{H}}}
		{{#1}}}}
	}
}
\MakeRobust{\SomeFormula}
\newcommand{\TheFormula}{{\SomeFormula{2}}}
\newcommand{\TheSafetyFormula}{{\psi}}

\newcommand{\After}[1]{{\left[#1\right]}}%{\bigcirc}%
\newcommand{\StrongAfter}[1]{{\left<#1\right>}}%{\bigcirc}%
\newcommand{\Globally}{{\square}}%{{\mathbf{G}}}
\newcommand{\Always}{{\Globally}}%{{\mathbf{G}}}
\newcommand{\Finally}{{\lozenge}}%{{\mathbf{F}}}
\newcommand{\Eventually}{{\Finally}}%{{\mathbf{G}}}

\newcommand{\SomeRelation}[1]{
	{\ifthenelse{\equal{#1}{1}}
		{{R}}
		{\ifthenelse{\equal{#1}{2}}
		{{S}}
		{\ifthenelse{\equal{#1}{3}}
		{{T}}
		{{#1}}}}
	}
}
\MakeRobust{\SomeRelation}
\newcommand{\TheRelation}{{\SomeRelation{1}}}

\newcommand{\SomeFunction}[1]{
	{\ifthenelse{\equal{#1}{1}}
		{f}
		{\ifthenelse{\equal{#1}{2}}
		{g}
		{\ifthenelse{\equal{#1}{3}}
		{h}
		{{#1}}}}
	}
}
\MakeRobust{\SomeFunction}
\newcommand{\TheFunction}{\SomeFunction{1}}
\newcommand{\OtherFunction}{\SomeFunction{2}}

\newcommand{\SomeCategory}[1]{
	{\ifthenelse{\equal{#1}{1}}
		{{\mathbf{C}}}
		{\ifthenelse{\equal{#1}{2}}
		{{\mathbf{D}}}
		{\ifthenelse{\equal{#1}{3}}
		{{\mathbf{E}}}
		{{\mathbf{#1}}}}}
	}
}
\MakeRobust{\SomeCategory}

\newcommand{\AsFunctor}[1]{\mathpzc{#1}}
\newcommand{\SomeFunctor}[1]{
	{
	\ifthenelse{\equal{#1}{1}}
		{\AsFunctor{F}}
		{\ifthenelse{\equal{#1}{2}}
		{\AsFunctor{G}}
		{\ifthenelse{\equal{#1}{3}}
		{\AsFunctor{H}}
		{\AsFunctor{{#1}}}}}
	}
}
\MakeRobust{\SomeFunctor}
\newcommand{\TheFunctor}{{\SomeFunctor{1}}}
\newcommand{\Deterministic}{{\SomeFunctor{D}}}
\newcommand{\NonDeterministic}{{\SomeFunctor{N}}}

\newcommand{\SomeMonad}[1]{
	{
	\ifthenelse{\equal{#1}{1}}
		{\AsFunctor{T}}
		{\ifthenelse{\equal{#1}{2}}
		{\AsFunctor{U}}
		{\ifthenelse{\equal{#1}{3}}
		{\AsFunctor{V}}
		{\AsFunctor{{#1}}}}}
	}
}
\MakeRobust{\SomeMonad}

\newcommand{\AsAlgebra}[1]{
	{\mathfrak{{#1}}}
}
\newcommand{\SomeAlgebra}[1]
{
	{\ifthenelse{\equal{#1}{1}}
	{\AsAlgebra{A}}
	{\ifthenelse{\equal{#1}{2}}
		{\AsAlgebra{B}}
		{\ifthenelse{\equal{#1}{3}}
			{\AsAlgebra{C}}
			{\PackageWarning{Preamble}{non-standard algebra symbol}#1}
		}
	}
	}
}
\MakeRobust{\SomeAlgebra}

\newcommand{\AsCoalgebra}[1]{
	{\mathbb{{#1}}}
}
\newcommand{\SomeCoalgebra}[1]{
	{\AsCoalgebra{\MakeUppercase{\SomeSet{#1}}}}
}
\MakeRobust{\SomeCoalgebra}
\newcommand{\TheCoalgebra}{{\SomeCoalgebra{1}}}
\newcommand{\AnotherCoalgebra}{{\SomeCoalgebra{2}}}

\newcommand{\SomeInput}[1]
{
	{\ifthenelse{\equal{#1}{1}}
	{i}%{\mathtt{a}}
	{\ifthenelse{\equal{#1}{2}}
		{j}%{b}
		{\ifthenelse{\equal{#1}{3}}
			{k}%{c}
			{\PackageWarning{Preamble}{non-standard input symbol}#1}
		}
	}
	}
}
\MakeRobust{\SomeInput}
\newcommand{\TheInput}{{\SomeInput{1}}}

\newcommand{\SomeSetOfInputs}[1]
{
	\ifthenelse{\equal{#1}{1}}
	{{I}}%{{A}}
	{\ifthenelse{\equal{#1}{2}}
		{B}
		{\ifthenelse{\equal{#1}{3}}
			{C}
			{\PackageWarning{Preamble}{non-standard set of inputs symbol}#1}
		}
	}
}
\MakeRobust{\SomeSetOfInputs}
\newcommand{\TheSetOfInputs}{{\SomeSetOfInputs{1}}}

\newcommand{\SomeSequenceOfInputs}[1]{
{
	{\ifthenelse{\equal{#1}{1}}
	{w}
	{\ifthenelse{\equal{#1}{2}}
		{u}
		{\ifthenelse{\equal{#1}{3}}
			{v}
			{\PackageWarning{Preamble}{non-standard sequence of inputs symbol}#1}
		}
	}
}}
}
\MakeRobust{\SomeSequenceOfInputs}

\newcommand{\TheSequenceOfInputs}{{\SomeSequenceOfInputs{1}}}

\newcommand{\TheInputSequence}{\TheSequenceOfInputs}

\newcommand{\SomeSetOfSequencesOfInputs}[1]{\SomeSetOfInputs{#1}^{*}}
\MakeRobust{\SomeSetOfSequencesOfInputs}

\newcommand{\SomeStreamOfInputs}[1]{
	\ifthenelse{\equal{#1}{1}}
		{\sigma}
		{\ifthenelse{\equal{#1}{2}}
		{\tau}
		{\ifthenelse{\equal{#1}{3}}
		{\rho}
		{{#1}}}}
	}
\MakeRobust{\SomeStreamOfInputs}

\newcommand{\SomeSetOfStreamsOfInputs}[1]{\SomeSetOfInputs{#1}^{\nat}}
\MakeRobust{\SomeSetOfStreamsOfInputs}

\newcommand{\SomeOutput}[1]
{
	\ifthenelse{\equal{#1}{1}}
	{o}
	{\ifthenelse{\equal{#1}{2}}
		{p}
		{\ifthenelse{\equal{#1}{3}}
			{q}
			{\PackageWarning{Preamble}{non-standard output symbol}#1}
		}
	}
}
\MakeRobust{\SomeOutput}
 
\newcommand{\TheOutputOf}[1]{{\TheOutputFunction(#1)}}%{#1\!\downarrow} 
\newcommand{\TheOutputFunction}{{\theta}}%{#1\!\downarrow} 
\newcommand{\TheOutputIs}[1]{\Observation\left(#1\right)}%{\vdash\!#1}%{{\Observation#1}}
\newcommand{\SomeSetOfOutputs}[1]
{
	\ifthenelse{\equal{#1}{1}}
	{{O}}
	{\ifthenelse{\equal{#1}{2}}
		{P}
		{\ifthenelse{\equal{#1}{3}}
			{Q}
			{\PackageWarning{Preamble}{non-standard set of outputs symbol}#1}
		}
	}
}
\MakeRobust{\SomeSetOfOutputs}
\newcommand{\TheSetOfOutputs}{{\SomeSetOfOutputs{1}}}

\newcommand{\SomeSequenceOfOutputs}[1]
{
	\ifthenelse{\equal{#1}{1}}
	{\underline{w}}
	{\ifthenelse{\equal{#1}{2}}
		{\underline{u}}
		{\ifthenelse{\equal{#1}{3}}
			{\underline{v}}
			{\PackageWarning{Preamble}{non-standard sequence of outputs symbol} \underline{#1}}
		}
	}
}
\MakeRobust{\SomeSequenceOfOutputs}

\newcommand{\SomeSetOfSequencesOfOutputs}[1]{\SomeSetOfOutputs{#1}^{*}}
\MakeRobust{\SomeSetOfSequencesOfOutputs}

\newcommand{\SomeException}[1]
{
	\ifthenelse{\equal{#1}{1}}
	{e}
	{\ifthenelse{\equal{#1}{2}}
		{d}
		{\ifthenelse{\equal{#1}{3}}
			{c}
			{\PackageWarning{Preamble}{non-standard Exception symbol}#1}
		}
	}
}
\MakeRobust{\SomeException}

\newcommand{\SomeSetOfExceptions}[1]
{
	\ifthenelse{\equal{#1}{1}}
	{E}
	{\ifthenelse{\equal{#1}{2}}
		{F}
		{\ifthenelse{\equal{#1}{3}}
			{G}
			{\PackageWarning{Preamble}{non-standard set of exceptions symbol}#1}
		}
	}
}
\MakeRobust{\SomeSetOfExceptions}

\newcommand{\AsSemiring}[1]{\AsCoalgebra{#1}}
\newcommand{\SomeSemiring}[1]{\AsSemiring{\MakeUppercase{\SomeSet{#1}}}}
\MakeRobust{\SomeSemiring}

\newcommand{\SomeFPS}[1]{
	\ifthenelse{\equal{#1}{1}}
		{\sigma}
		{\ifthenelse{\equal{#1}{2}}
		{\gamma}
		{\ifthenelse{\equal{#1}{3}}
		{\rho}
		{{#1}}}}
	}
\MakeRobust{\SomeFPS}

\newcommand{\SomeBehaviour}[1]{{\SomeFPS{#1}}}
\MakeRobust{\SomeBehaviour}
\newcommand{\TheBehaviour}{{\SomeBehaviour{1}}}

\newcommand{\TheBehaviourOf}[1]{{{\llbracket#1\rrbracket}}}%{{{#1}_{beh}}}

\newcommand{\AsSequence}[1]{{{\left(#1\right)}}}%{{{\left<#1\right>}}}
\usetikzlibrary{arrows,automata}

\makeatletter
\newcommand{\da@rightarrow}{\mathchar"0\hexnumber@\symAMSa 4B }
\newcommand{\da@leftarrow}{\mathchar"0\hexnumber@\symAMSa 4C }
\newcommand{\xdashedrightarrow}[2][]{%
  \mathrel{%
    \mathpalette{\da@xarrow{#1}{#2}{}\da@rightarrow{\,}{}}{}%
  }%
}
\newcommand{\xdashedleftarrow}[2][]{%
  \mathrel{%
    \mathpalette{\da@xarrow{#1}{#2}\da@leftarrow{}{}{\,}}{}%
  }%
}
\newcommand{\da@xarrow}[7]{%
  \sbox0{$\ifx#7\scriptstyle\scriptscriptstyle\else\scriptstyle\fi#5#1#6\m@th$}%
  \sbox2{$\ifx#7\scriptstyle\scriptscriptstyle\else\scriptstyle\fi#5#2#6\m@th$}%
  \sbox4{$#7\dabar@\m@th$}%
  \dimen@=\wd0 %
  \ifdim\wd2 >\dimen@
    \dimen@=\wd2 %   
  \fi
  \count@=2 %
  \def\da@bars{\dabar@\dabar@}%
  \@whiledim\count@\wd4<\dimen@\do{%
    \advance\count@\@ne
    \expandafter\def\expandafter\da@bars\expandafter{%
      \da@bars
      \dabar@ 
    }%
  }%  
  \mathrel{#3}%
  \mathrel{%   
    \mathop{\da@bars}\limits
    \ifx\\#1\\%
    \else
      _{\copy0}%
    \fi
    \ifx\\#2\\%
    \else
      ^{\copy2}%
    \fi
  }%   
  \mathrel{#4}%
}

\makeatletter
\providecommand*{\twoheadrightarrowfill@}{%
  \arrowfill@\relbar\relbar\twoheadrightarrow
}
\providecommand*{\twoheadleftarrowfill@}{%
  \arrowfill@\twoheadleftarrow\relbar\relbar
}
\providecommand*{\xtwoheadrightarrow}[2][]{%
  \ext@arrow 0579\twoheadrightarrowfill@{#1}{#2}%
}
\providecommand*{\xtwoheadleftarrow}[2][]{%
  \ext@arrow 5097\twoheadleftarrowfill@{#1}{#2}%
}
\makeatother

\newcolumntype{L}{> {$}l <{$} } % math-mode version of "l" column type
\newcolumntype{R}{> {$}r <{$} } % math-mode version of "r" column type

\begin{document}
\maketitle
\begin{abstract}
This work aims to solve a practical problem, i.e., how to quantify the risk brought upon a system by different attackers. The answer is useful for optimising resource allocation for system defence. Given a set of safety requirements, we quantify the attacker capability in terms of the set of safety requirements an attacker can compromise. Given a system (in the presence of an attacker), model checking it against each safety requirement one by  one is expensive and wasteful (since the same state space is explored many times). We thus propose model checking multiple properties efficiently by means of coalgebraic model checking using {enhanced coinduction} techniques.  We apply the proposed technique to a real-world water treatment system and the results show that our approach can effectively reduce the effort required for model checking. 
\end{abstract}

\section{Introduction}
Having an understanding of the capabilities of attackers helps us decide how to use our security resources in the most efficient way. Consider, for example, a water treatment plant; it is unwise to put all efforts on physical security if an attacker can hack into the controller through the internet and can then damage the physical components by wrongful manipulation of the controller. 

Attacker models often include a list of capabilities available to the attacker that allow her to interact and interfere with the target system. Basin and Cremers \cite{KnowYourEnemy} classify attackers (each with an associated list of adversary actions) based on whether each attacker can violate a given security property during a protocol run using her available actions. Our view is that the capabilities of powerful attackers are ultimately related to  failed security requirements (e.g., an attacker that can reveal the state of the system violates the property that the state of the system has to remain secret). In this sense, given a set of basic security requirements, we could classify attackers based on how many of these properties are violated.

In the case of CPSs, safety requirements describe which ranges are considered safe for the different elements of the plant, usually in the form of invariants (e.g, an invariant that defines the safe value readings for a pressure sensor). For a given attacker and a given set of safety properties, we can model check each safety property on the system under attack to quantify the attacker. However, this approach is wasteful since we explore the state space many times (once per property), and checking the product of all properties does not explicitly separate those properties that are violated from those that are not. Moreover, for CPSs, some of these safety properties are related by physical relationships (e.g., the pressure inside a water tank is related to the water level). In these cases, we would like to automatically infer the satisfaction/failure of related properties from previous verifications with the goal to reduce the effort required to check multiple properties. 

The following research questions motivate our work: \textbf{(RQ1)} how can we reduce the effort of verifying multiple properties for the quantification of attacker capabilities? and \textbf{(RQ2)} how do we quantify the attackers using the verification results? To answer RQ1, we need to show that the number of states or time which is required to verify the properties one by one can be reduced when verifying multiple properties, and to answer RQ2, we show the comparison between these attackers and justify that the results are meaningful. 

Our contributions are: 
$\bullet$ A formalisation of the problem of attacker quantification from the perspective of model checking multiple properties. 
$\bullet$ A model checking algorithm for safety properties which uses enhanced coinduction \cite{EnhancedCoalgebraicBisimulation,EnhancedCoinduction} to speed up the verification of multiple properties; this speed-up is the result of both exploiting algebraic properties of the state space, and by reusing knowledge obtained in previous verifications.

\section{Preliminaries}
\label{sec:preliminaries}
In this section, we provide the formalisms necessary to precisely present the problem of attacker quantification.

A \emph{deterministic system} with inputs in a set $\TheSetOfInputs$ and observations in a set $\TheSetOfOutputs$ is a tuple $(\TheSet,\TheOutputFunction_\Deterministic,\delta_\Deterministic)$, where $\TheSet$ is the set of states equipped with an \emph{observation} function $\TheOutputFunction_\Deterministic\colon \TheSet \rightarrow {\TheSetOfOutputs}$ and a \emph{transition} function $\delta_\Deterministic\colon\TheSet \times\TheSetOfInputs\rightarrow \TheSet$. Usually, deterministic systems have a distinguished initial state $\TheElement_0\in \TheSet$. A \emph{nondeterministic system} with inputs in $\TheSetOfInputs$ and observations in $\TheSetOfOutputs$ is a tuple $(\TheSet,\TheOutputFunction_\NonDeterministic,\delta_\NonDeterministic)$ where, again, $\TheSet$ is a set of states equipped with an {observation} function $\TheOutputFunction_\NonDeterministic\colon \TheSet \rightarrow {\TheSetOfOutputs}$, but this time the {transition} function has the signature $\delta_\NonDeterministic\colon\TheSet\times\TheSetOfInputs \rightarrow \TheFinitePowersetOf{\TheSet}$, where $\TheFinitePowersetOf{\TheSet}$ is the set of finite subsets of $\TheSet$. Nondeterministic systems usually have a finite set of initial states $\TheSet_0\subseteq\TheSet$ instead of a single initial state. 

\emph{$\TheFunctor$-coalgebras}  are a general framework which allows us to model both deterministic and non-deterministic systems as systems of the same type \cite{UniversalCoalgebra}. 
\begin{definition}[$\TheFunctor$-coalgebras]
\label{def:FCoalgebra}
An $\TheFunctor$-\emph{coalgebra} $\TheCoalgebra$ is a tuple $(\TheSet,\TheOutputFunction,\delta)$ where $\TheSet$ is the set of states, equipped with an \emph{observation} function $\TheOutputFunction\colon \TheSet \rightarrow \ThePowersetOf{\TheSetOfOutputs}$ and a {transition} function $\delta\colon\TheSet \rightarrow \TheSet^\TheSetOfInputs$. 
\end{definition}
For a deterministic system $(\TheSet,\TheOutputFunction_\Deterministic,\delta_\Deterministic)$, its corresponding $\TheFunctor$-coalgebra is $(\TheSet,\TheOutputFunction,\delta)$, where $\TheOutputOf{\TheElement}\triangleq\set{\TheOutputFunction_\Deterministic(\TheElement)}$ and $\delta\triangleq \delta_\Deterministic$. For a nondeterministic system $(\TheSet,\TheOutputFunction_\NonDeterministic,\delta_\NonDeterministic)$, its corresponding $\TheFunctor$-coalgebra is $(\TheFinitePowersetOf{\TheSet},\TheOutputFunction,\delta)$, where $\TheOutputFunction(\OtherSet)\triangleq\set{\TheOutputFunction_\Deterministic(\TheElement)|\TheElement\in \OtherSet}$ and $\delta(\OtherSet)\triangleq\cup\set{\delta_\NonDeterministic(\TheElement)|\TheElement\in \OtherSet}$, for $\OtherSet\subseteq \TheSet$.

The transition function $\delta$ of $\TheFunctor$-coalgebras can be iterated to yield the function $\delta^{*}\colon\TheSet\rightarrow \TheSet^{\TheSetOfInputs^{*}}$, which is inductively defined, for $\TheElement\in \TheSet$, $\TheInput\in \TheSetOfInputs$ and $\TheSequenceOfInputs\in \TheSetOfInputs^{*}$, by
\begin{align}
\delta^{*}(\TheElement)(\varepsilon)\triangleq \TheElement,\quad~\text{and}~\quad
\delta^{*}(\TheElement)(\TheInput:\TheSequenceOfInputs)\triangleq \delta^{*}(\delta(\TheElement)(\TheInput))(\TheSequenceOfInputs),
\end{align}
where $\varepsilon$ is the empty sequence, and $:$ prepends $\TheInput$ to the sequence $\TheInputSequence$. 

\begin{definition}[Observable Behaviours]
\label{def:ObservableBehaviour}
Let $\ThePowersetOf{\TheSetOfOutputs}^{\TheSetOfInputs^{*}}$ be the set of functions that map a finite sequence of inputs $\TheInputSequence \in \TheSetOfInputs^{*}$ to a set of observations. Given an $\TheFunctor$-coalgebra $\TheCoalgebra=\AsSequence{\TheSet,\theta,\delta}$, the mapping $\TheBehaviourOf{-}_\TheCoalgebra\colon \TheSet\rightarrow \ThePowersetOf{\TheSetOfOutputs}^{\TheSetOfInputs^{*}}$ maps $\TheElement\in \TheSet$ to its \emph{observable behaviour} $\TheBehaviourOf{\TheElement}_\TheCoalgebra\in\ThePowersetOf{\TheSetOfOutputs}^{\TheSetOfInputs^{*}}$, defined, for $\TheInput \in \TheSetOfInputs$ and $\TheInputSequence \in \TheSetOfInputs^{*}$, by
\begin{align}
\TheBehaviourOf{\TheElement}_\TheCoalgebra(\varepsilon)\triangleq \TheOutputOf{\TheElement},\quad\text{ and }\quad
\TheBehaviourOf{x}_\TheCoalgebra(i:\TheInputSequence)\triangleq \TheBehaviourOf{\delta(x)(i)}_\TheCoalgebra(\TheInputSequence).
\end{align}
We usually refer to $\TheBehaviourOf{-}_\TheCoalgebra$ simply as $\TheBehaviourOf{-}$ when the $\TheFunctor$-coalgebra $\TheCoalgebra$ is clear from the context.
\end{definition}
We quantify the capabilities of the attacker by determining which security properties can be violated by using the attacks available to him/her. We model attacks as functions that change the \emph{behaviours} of an $\TheFunctor$-coalgebra. 
\begin{definition}[Attacks and Attackers]
Any function $\alpha\colon \ThePowersetOf{\TheSetOfOutputs}^{\TheSetOfInputs^{*}} \rightarrow \ThePowersetOf{\TheSetOfOutputs}^{\TheSetOfInputs^{*}}$ is, by definition, an \emph{attack}. An $\TheFunctor$-coalgebra $\TheCoalgebra=(\TheSet, \TheOutputFunction,\delta)$ affected by $\alpha$ has the observable behaviour $\TheBehaviourOf{\TheElement}$ of each element $\TheElement\in \TheSet$ replaced by the behaviour $\alpha(\TheBehaviourOf{\TheElement})$. %The set $\left(\ThePowersetOf{\TheSetOfOutputs}^{\TheSetOfInputs^{*}}\right)^{\left(\ThePowersetOf{\TheSetOfOutputs}^{\TheSetOfInputs^{*}}\right)}$ is the set of all possible attacks. 
An \emph{attacker} $A$ is a set of attacks.%; i.e., $A\subseteq \left(\ThePowersetOf{\TheSetOfOutputs}^{\TheSetOfInputs^{*}}\right)^{\left(\ThePowersetOf{\TheSetOfOutputs}^{\TheSetOfInputs^{*}}\right)}$.
\end{definition}

The set $\ThePowersetOf{\TheSetOfOutputs}^{\TheSetOfInputs^{*}}$ has a \emph{final $\TheFunctor$-coalgebra structure} \cite{UniversalCoalgebra}, which enables the \emph{coinductive definition} of attacks. 
\begin{definition}[Final $\TheFunctor$-coalgebra]
The $\TheFunctor$-coalgebra $\Omega\triangleq(\ThePowersetOf{\TheSetOfOutputs}^{\TheSetOfInputs^{*}}, \theta_\Omega, \delta_\Omega)$ is the \emph{final $\TheFunctor$-coalgebra}; defined for $\sigma \in \ThePowersetOf{\TheSetOfOutputs}^{\TheSetOfInputs^{*}}$, $\TheInput\in\TheSetOfInputs$, and $\TheSequenceOfInputs\in \TheSetOfInputs^{*}$ by
\begin{align}\theta_\Omega(\sigma)\triangleq \sigma(\varepsilon),\quad\text{ and }\quad\delta_\Omega(\sigma)(\TheInput)\triangleq \lambda \TheSequenceOfInputs \in \TheSetOfInputs^{*} \ldotp\sigma(\TheInput:\TheSequenceOfInputs).\end{align}
\end{definition}
The commutative diagram in Figure~\ref{fig:AttackedBehaviour} illustrates the effect of an attack $\alpha$ on the elements of an $\TheFunctor$-coalgebra $\TheCoalgebra=(\TheSet, \theta,\delta)$. 
Because $\Omega$ is the final $\TheFunctor$-coalgebra, the pairing $(\theta_\Omega, \delta_\Omega)$ is an isomorphism \cite{UniversalCoalgebra}, and we can coinductively define the attack $\alpha$ through the functions $\theta_\alpha\colon \ThePowersetOf{\TheSetOfOutputs}^{\TheSetOfInputs^{*}}\rightarrow  \ThePowersetOf{\TheSetOfOutputs} $ and $\delta_\alpha\colon \ThePowersetOf{\TheSetOfOutputs}^{\TheSetOfInputs^{*}} \rightarrow \left(\ThePowersetOf{\TheSetOfOutputs}^{\TheSetOfInputs^{*}}\right)^\TheSetOfInputs$, since $\alpha=(\theta_\Omega,\delta_\Omega)^{-1}\circ(\theta_\alpha,\delta_\alpha)$, where $\circ$ is function composition.
{We take a coinductive approach for the definition of attacks because coinductive definitions make explicit the effects of an attack over observations and transitions.}

\begin{figure}[t]
\centering
{\scriptsize
{
\begin{tikzpicture}[descr/.style={fill=white,inner sep=2.5pt}]
  \matrix (m) [matrix of math nodes, row sep=3em,
  column sep=7em]
  {
  X & \ThePowersetOf{\TheSetOfOutputs}^{\TheSetOfInputs^{*}}& \ThePowersetOf{\TheSetOfOutputs}^{\TheSetOfInputs^{*}}%& \ThePowersetOf{\TheSetOfOutputs}^{\TheSetOfInputs^{*}}
  \\
  \ThePowersetOf{\TheSetOfOutputs}\times \TheSet^\TheSetOfInputs & \ThePowersetOf{\TheSetOfOutputs}\times \left(\ThePowersetOf{\TheSetOfOutputs}^{\TheSetOfInputs^{*}}\right)^\TheSetOfInputs& \ThePowersetOf{\TheSetOfOutputs}\times \left(\ThePowersetOf{\TheSetOfOutputs}^{\TheSetOfInputs^{*}}\right)^I%& O\times \ThePowersetOf{\TheSetOfOutputs}^{\TheSetOfInputs^{*}}
  \\
  };
  \path[->]
	(m-1-1) edge node[descr] {$ (\theta,\delta)$} (m-2-1) 
	(m-1-2) edge node[descr] {$ (\theta_\Omega,\delta_\Omega)$} (m-2-2) edge node[auto] {$ \alpha $} (m-1-3) 
	(m-2-3) edge node[descr] {$ (\theta_\Omega,\delta_\Omega)^{-1}$} %node [left] {$\cong$} 
	(m-1-3)
	(m-1-2) edge node[descr] {$ (\theta_\alpha,\delta_\alpha)$} (m-2-3)
	%(m-1-4) edge node[descr] {$ (\theta_\Omega,\delta_\Omega)$} (m-2-4)
	%(m-2-2) edge node[auto] {$ id_O\times\alpha^{id_I} $} (m-2-3)  %This arrow is not correct because $\alpha$ cannot be a homomorphism (unless it is the identity)
 ;
 \path[->, dashed]
	(m-1-1) edge node[auto] {$\TheBehaviourOf{-}$} (m-1-2) 
	%(m-1-3) edge node[auto] {$\TheBehaviourOf{-}$} (m-1-4) 
	(m-2-1) edge node[auto] {$\id_{ \ThePowersetOf{\TheSetOfOutputs}}\times \TheBehaviourOf{-}^{\id_I}$} (m-2-2)
	%(m-2-3) edge node[auto] {$id_O\times \TheBehaviourOf{-}$} (m-2-4)
 ;
%\path[right hook->]
%	(m-2-1) edge node[auto, swap] {$ \AsFunction{emb}_2 $} (m-2-2) 
%	(m-2-3) edge node[below]  {$ \AsFunction{emb}_1\times id_Y $} (m-2-5);
\end{tikzpicture}}}
\caption{Transformation caused by an attack $\alpha\colon \ThePowersetOf{\TheSetOfOutputs}^{\TheSetOfInputs^{*}}\rightarrow \ThePowersetOf{\TheSetOfOutputs}^{\TheSetOfInputs^{*}}$ on the behaviours of an $\TheFunctor$-coalgebra $(\TheSet, \theta,\delta).$ (A function $\id_{\ThePowersetOf{\TheSetOfOutputs}}\times\TheFunction^{\id_{\TheSetOfInputs}}$ is defined, for $(\AnotherSet,\OtherFunction)\in  \ThePowersetOf{\TheSetOfOutputs}\times \TheSet^\TheSetOfInputs$, by 
$
\left(\id_{\ThePowersetOf{\TheSetOfOutputs}}\times\TheFunction^{\id_{\TheSetOfInputs}}\right)(\AnotherSet,\OtherFunction)\triangleq(\AnotherSet,\TheFunction\circ\OtherFunction))$, where $\circ$ is function composition.
}
\label{fig:AttackedBehaviour}
\vspace{-0.25cm}
\end{figure}
\begin{example}[Attacking a Dial of a Combination Lock]
\label{ex:Dial}
Let $\AsSet{10}=\set{0,..,9}$, $\TheSetOfInputs=\set{\bullet}$ and $\TheSetOfOutputs=\AsSet{10}$. We model a \emph{dial system} using the $\TheFunctor$-coalgebra $\mathbb{D}=(\AsSet{10}, \TheOutputFunction,\delta)$, defined for $\TheElement\in\AsSet{10} $ by:
\begin{align*}
\TheOutputOf{\TheElement}\triangleq\set{\TheElement},\quad\text{and}\quad \delta(\TheElement,\bullet)\triangleq\TheElement+1~(mod~10)
\end{align*}
We coinductively define an attack $\alpha$, for $\TheBehaviour\in \ThePowersetOf{\TheSetOfOutputs}^{\TheSetOfInputs^{*}}$ and $\TheInput\in\TheSetOfInputs$, by $\theta_\alpha(\TheBehaviour)\triangleq\set{0}$ and $\delta_\alpha(\TheBehaviour)(\TheInput)\triangleq\alpha(\delta_\Omega(\TheBehaviour)(\TheInput))$. It is not difficult to show, for $\sigma\in\ThePowersetOf{\TheSetOfOutputs}^{\TheSetOfInputs^{*}}$ and $\TheInputSequence \in \TheSetOfInputs^{*}$, that $\alpha(\sigma)(\TheInputSequence)=\set{0}$. The attack $\alpha$ models a change in the observation function $\theta$ such that $\theta(\TheElement)$ evaluates to $\set{0}$ for all $\TheElement\in \AsSet{10}$. 

We also coinductively define an attack $\beta$ by $\theta_\beta(\TheBehaviour)\triangleq\theta_\Omega(\TheBehaviour)$ and $
\delta_\beta(\TheBehaviour)(\TheInput)\triangleq \beta(\TheBehaviourOf{0})$. The attack $\beta$ is slightly different from the attack $\alpha$, since it models a change in the transition function $\delta$ such that $\delta(\TheElement)$ evaluates to $0$ for all $\TheElement\in \AsSet{10}$. Both attacks have the same effect on the behaviour $\TheBehaviourOf{0}$, i.e., $\alpha(\TheBehaviourOf{0})=\beta(\TheBehaviourOf{0})$; thus, we see that different attacks can have the same effect over the behaviour of some states. Finally, from the set of attacks $\set{\alpha,\beta}$ we can obtain four attackers, which correspond to the subsets of $\set{\alpha,\beta}$.
\end{example}
We now present a logic to describe properties of $\TheFunctor$-coalgebras. Coalgebraic logics are $\mu$-calculus-like logics that naturally define an $\TheFunctor$-coalgebra on formulae. This provides intuitive semantics for formulae (a $\TheFunctor$-coalgebra) and formulae satisfaction (in terms of relationships between $\TheFunctor$-coalgebras).
\begin{definition}[Safety Formulae] 
\label{def:SafetyFormulae}
Let $\TheSetOfVariables$ be a set of variables. Let $\ThePredicate\subseteq \TheSetOfInputs$, $\OtherPredicate \subseteq \TheSetOfOutputs$, and $\TheVariable \in \TheSetOfVariables$; the set $\TheSetOfSafetyFormulae{{}}$ of \emph{safety formulae} is given by the BNF syntax
\begin{align}
\TheSafetyFormula &::=%\True\ |\ 
\TheVariable\ |\ \After{\ThePredicate}\TheSafetyFormula\ |
\ \TheOutputIs{\OtherPredicate}
\ 
|\  \TheSafetyFormula \land \TheSafetyFormula\
|\ \ \nu\TheVariable.\TheSafetyFormula(\TheVariable).
\end{align}
\end{definition}
Formulae of the form $\After{\ThePredicate}\TheSafetyFormula$ are \emph{transition formulae}, formulae of the form $\TheOutputIs{\OtherPredicate}$ are \emph{observation formulae}, and formulae of the form $\TheSafetyFormula(\TheVariable)$ are \emph{guarded modal formulae}. 
We denote the set of closed and guarded safety formulae (i.e., formulae without free occurrences of fixed point variables $\TheVariable$, unless guarded) by $\TheSetOfClosedAndGuardedSafetyFormulae{{}}$. We also use the syntactic sugar $\True$ to denote $\TheOutputIs{\TheSetOfOutputs}$, and $\lnot \TheOutputIs{\OtherPredicate}$ to denote $\TheOutputIs{\TheSetOfOutputs - \OtherPredicate}$. 

The semantics of a safety formula describes which observations are valid for the current state, and which obligations must hold for successor states so that the formula is valid at the current state. We define these semantics using an $\TheFunctor$-coalgebra.

\begin{definition}[Semantics] 
\label{def:SemanticsOfDTSFormulae}
We define the $\TheFunctor$-coalgebra $\zeta$
by $\zeta\triangleq(\TheSetOfClosedAndGuardedSafetyFormulae, \TheOutputFunction_\zeta, \delta_\zeta),$
where $ \TheOutputFunction_\zeta\colon\TheSetOfClosedAndGuardedSafetyFormulae\rightarrow \ThePowersetOf{O}\text{ and } \delta_\zeta\colon\TheSetOfClosedAndGuardedSafetyFormulae\rightarrow\left(\TheSetOfClosedAndGuardedSafetyFormulae\right)^\TheSetOfInputs$ are defined in Table~\ref{tab:Semantics}. 
\begin{table}[!t]
\begin{mdframed}
\begin{tabularx}{0.5\textwidth}{L  L L L L L L }
\TheOutputFunction_\zeta(\After{\ThePredicate}\TheSafetyFormula)	&\triangleq&	O &\quad&\delta_\zeta(\After{\ThePredicate}\TheSafetyFormula,i)	&\triangleq&
	\begin{cases}	
		\TheSafetyFormula,&\text{if $i\in \ThePredicate$}\\
		\True, &\text{otherwise.}
	\end{cases}\\
\TheOutputFunction_\zeta(\TheOutputIs{\OtherPredicate})	&\triangleq&	\OtherPredicate &\quad&\delta_\zeta(\TheOutputIs{\OtherPredicate},i)	&\triangleq&	\True\\
\TheOutputFunction_\zeta(\TheSafetyFormula_1 \land \TheSafetyFormula_2)	&\triangleq&	\TheOutputFunction_\zeta(\TheSafetyFormula_1) \cap \TheOutputFunction_\zeta(\TheSafetyFormula_2)&\quad\quad\quad&\delta_\zeta(\TheSafetyFormula_1 \land \TheSafetyFormula_2)(i)	&\triangleq&\delta_\zeta(\TheSafetyFormula_1)(i) \land \delta_\zeta(\TheSafetyFormula_2)(i)\\
\TheOutputFunction_\zeta(\nu\TheVariable.\TheSafetyFormula(\TheVariable))	&\triangleq&	\TheOutputFunction_\zeta(\TheSafetyFormula\left[\nu\TheVariable.\TheSafetyFormula(\TheVariable)/\TheVariable\right])&\quad&\delta_\zeta(\nu\TheVariable.\TheSafetyFormula(\TheVariable))(\TheInput)	&\triangleq&\delta_\zeta(\TheSafetyFormula\left[\nu\TheVariable.\TheSafetyFormula(\TheVariable)/\TheVariable\right])(i)
\end{tabularx}
\end{mdframed}
\caption{Definitions of $\TheOutputFunction_\zeta$ and $\delta_\zeta$. The expression $\TheSafetyFormula\left[\nu\TheVariable.\TheSafetyFormula(\TheVariable)/\TheVariable\right]$ denotes the syntactic substitution of every free occurrence of $\TheVariable$ in $\TheSafetyFormula(\TheVariable)$ by $\nu\TheVariable.\TheSafetyFormula(\TheVariable)$. }
\label{tab:Semantics}
\vspace{-0.75cm}
\end{table}

Since the only $\TheFunctor$-coalgebra we define over $\TheSetOfClosedAndGuardedSafetyFormulae$ is $\zeta$, we henceforth write $\theta(\TheSafetyFormula)$ and $\delta(\TheSafetyFormula)$ instead of $\theta_\zeta(\TheSafetyFormula)$ and $\delta_\zeta(\TheSafetyFormula)$ to ease readability. 
\end{definition}

The size of $\TheSafetyFormula$ is the number of formulae that are reachable from it, i.e.,
\begin{align*}
|\TheVariable|~&= 1,\quad\quad |~[P]\TheSafetyFormula| = 1+|\TheSafetyFormula|,\quad \quad| \TheOutputIs{Q} |= 1,\\
|\TheSafetyFormula_1\land\TheSafetyFormula_2|~&= max(|\TheSafetyFormula_1|,|\TheSafetyFormula_2|),\quad\quad\ \   |\nu\TheVariable.\TheSafetyFormula(\TheVariable)|~=|\TheSafetyFormula\left[\nu\TheVariable.\TheSafetyFormula(\TheVariable)/\TheVariable\right])|.
\end{align*}

To capture the notion of safety formulae satisfaction, we use \emph{simulation relations} between $\TheFunctor$-coalgebras.
\begin{definition}[Simulation of $\TheFunctor$-coalgebras]
\label{def:Simulation}
Let $\TheCoalgebra=\AsSequence{\TheSet,\theta_\TheCoalgebra, \delta_\TheCoalgebra}$ and $\AnotherCoalgebra=\AsSequence{\AnotherSet,\theta_\AnotherCoalgebra, \delta_\AnotherCoalgebra}$ be $\AsFunctor{F}$-coalgebras. Given $\TheRelation\subseteq \TheSet\times\AnotherSet$, we say that $\TheRelation$ is a \emph{simulation from $\TheCoalgebra$ to $\AnotherCoalgebra$} if and only if $
\theta_\TheCoalgebra(\TheElement)\subseteq\theta_\AnotherCoalgebra(\AnotherElement)$ and $(\delta_\TheCoalgebra(\TheElement)(\TheInput),\delta_\AnotherCoalgebra(\AnotherElement)(\TheInput)) \in \TheRelation$ for all $(\TheElement, \AnotherElement)\in \TheRelation$ and $\TheInput \in \TheSetOfInputs$.
The greatest simulation relation, %in $\TheCoalgebra$ and $\AnotherCoalgebra$, 
$\lesssim$, is called \emph{similarity}.
\end{definition}

\begin{definition}[Safety Formula Satisfaction]
\label{def:Satisfaction}
Given an $\TheFunctor$-coalgebra $\SomeCoalgebra{1}=(\TheSet, \TheOutputFunction, \delta)$, $\TheElement\in \TheSet$, and a safety formula $\TheSafetyFormula\in \TheSetOfClosedAndGuardedSafetyFormulae{{}}$, we say that $\TheElement$ \emph{satisfies} $\TheSafetyFormula$ in the context of $\TheCoalgebra$, written $(\TheCoalgebra,\TheElement)\vDash \TheSafetyFormula$, if and only if $\TheElement\lesssim\TheSafetyFormula$; i.e., the pair $(\TheElement,\TheSafetyFormula)$ belongs to the similarity relation from $\TheCoalgebra$ to $\zeta$.
\end{definition}

In other words, $(\TheCoalgebra,\TheElement)\vDash \TheSafetyFormula$ if and only if all the observations of $\TheElement$ are in the observations allowed by $\TheSafetyFormula$, and all the successors of $\TheElement$ satisfy the transition formulae associated with $\TheSafetyFormula$. More precisely, $(\TheCoalgebra,\TheElement)\vDash \TheSafetyFormula$, if and only \begin{align}
\TheOutputOf{\TheElement}\subseteq\TheOutputFunction(\TheSafetyFormula), \quad\text{and}\quad (\TheCoalgebra,\delta(\TheElement)(\TheInput))\vDash \delta(\TheSafetyFormula)(\TheInput),~\text{for all $\TheInput\in \TheSetOfInputs$.}\label{eq:Transition}
\end{align}
\begin{remark}[Formulae Implication]
\label{rem:FormulaeImplication}
Consider two safety formulae $\TheSafetyFormula_1$ and $\TheSafetyFormula_2$ with $\TheSafetyFormula_1\lesssim\TheSafetyFormula_2$. Whenever $(\TheCoalgebra,\TheElement)\vDash \TheSafetyFormula_1$, we can conclude $(\TheCoalgebra,\TheElement)\vDash \TheSafetyFormula_2$ that since the composition of simulation relations is a simulation relation itself. Thus, the similarity relation $\lesssim$ in the $\TheFunctor$-coalgebra $\zeta$ corresponds to \emph{formulae implication}.
\end{remark}

\begin{definition}[Formulae] 
\label{def:Formulae}
Let $\ThePredicate\subseteq \TheSetOfInputs$ and $\TheSafetyFormula \in \TheSetOfClosedAndGuardedSafetyFormulae$; the set $\TheSetOfFormulae{{}}$ of \emph{formulae} is given by the BNF syntax
\begin{align}
\TheFormula &::= \TheSafetyFormula\ | \ \lnot \TheFormula \ |\ 
\TheVariable\ |\ \After{\ThePredicate}\TheFormula\ |
\  \TheFormula \land \TheFormula\
|\ \ \nu\TheVariable.\TheFormula(\TheVariable).
\end{align}
\end{definition}
We use syntactic sugar in the usual way:  $\TheFormula_1\lor \TheFormula_2$ denotes $\lnot (\lnot \TheFormula_1 \land \lnot \TheFormula_2)$, $\TheFormula_1\Rightarrow \TheFormula_2$ denotes $\lnot \TheFormula_1 \lor \TheFormula_2$, $\StrongAfter{\ThePredicate}\TheFormula$ denotes $\lnot (\After{\ThePredicate} \lnot \TheFormula)$, $ \mu\TheVariable.\TheFormula(\TheVariable)$ denotes $ \lnot (\nu\TheVariable.\lnot \TheFormula(\lnot \TheVariable))$, $\Always \TheFormula$ denotes $\nu \TheVariable. \TheFormula \land \After{\TheSetOfInputs}(\TheVariable)$, and  $\Eventually \TheFormula$ denotes $\lnot  (\Always \lnot \TheFormula)$. Similarly, the set $\TheSetOfClosedAndGuardedFormulae$ of closed and guarded formulae contains all formulae that have no occurrences of free variables.

The notion of satisfaction for non-safety formulae extends that of safety formulae, and it is inductively defined as follows:
\begin{align*}
&(\TheCoalgebra, \TheElement) \vDash \lnot \TheFormula &&\Leftrightarrow&& \lnot ((\TheCoalgebra, \TheElement) \vDash  \TheFormula),\\
&(\TheCoalgebra, \TheElement) \vDash \After{\ThePredicate}\TheFormula &&\Leftrightarrow&&(\forall \TheInput \in \ThePredicate\ldotp(\TheCoalgebra, \delta(\TheElement)(\TheInput)) \vDash  \TheFormula),\\
&(\TheCoalgebra, \TheElement) \vDash \TheFormula_1 \land  \TheFormula_2 &&\Leftrightarrow&& (\TheCoalgebra, \TheElement) \vDash \TheFormula_1 \land (\TheCoalgebra, \TheElement) \vDash \TheFormula_2\\
&(\TheCoalgebra, \TheElement) \vDash  \nu\TheVariable.\TheFormula(\TheVariable) &&\Leftrightarrow&& (\TheCoalgebra, \TheElement) \vDash \TheFormula\left[\nu\TheVariable.\TheFormula(\TheVariable)/\TheVariable\right].
\end{align*}

\begin{definition}[Counterexample]
\label{def:Counterexamples}
Let $\SomeCoalgebra{1}=(\TheSet, \TheOutputFunction, \delta)$ be an $\TheFunctor$-coalgebra, $\TheElement\in \TheSet$, and $\TheFormula\in \TheSetOfClosedAndGuardedFormulae{{}}$; say that a pair $(\TheElement_{\TheCounterexample},\TheFormula_{\TheCounterexample})$ is a \emph{counterexample for $(\TheCoalgebra,\TheElement)\vDash \TheFormula$} 
 
if and only if $
(\TheCoalgebra,\TheElement_{\TheCounterexample})\not\vDash\TheFormula_{\TheCounterexample}$ and that $(\TheCoalgebra,\TheElement_{\TheCounterexample})\not\vDash\TheFormula_{\TheCounterexample}$ implies $(\TheCoalgebra,\TheElement)\not\vDash\TheFormula$.

\end{definition}
As our counterexamples do not involve state traces, they could be considered ``not very informative'' for the purposes of debugging or security patching; however, the existence of a counterexample suffices for the purposes of attacker quantification.

Following Equation~\ref{eq:Transition}, for a given a safety formula $\TheSafetyFormula$, there are two ordinary reasons for the existence of a counterexample for $(\TheCoalgebra,\TheElement)\vDash \TheSafetyFormula$: either $(\TheElement, \TheSafetyFormula)$ is a counterexample, or there is a counterexample $(\delta(\TheElement)(\TheInput),\delta(\TheSafetyFormula)(\TheInput))$ for some $\TheInput\in \TheSetOfInputs$. 
\begin{example}%[Checking a Dial System]
\label{ex:CheckDials}
Consider the predicate $(\cdot = z)\triangleq\set{y|y\in \OtherSet \text{~and~}y=z}$ for some set $\OtherSet$, and similarly for the predicate $(\cdot \neq z)$. Recall the dial system from Example~\ref{ex:Dial}, and consider the non-safety property $\Eventually\TheOutputIs{\cdot = n}$ at state 0, which informally means ``the dial eventually displays the value $n$ (for some $n \in \AsSet{10}$).'' The high-level verification procedure is as follows: 
since $\Eventually\TheOutputIs{\cdot = n} =\lnot\left(\Always\TheOutputIs{{\cdot\neq n}}\right)$, and since $(n,\Always{\TheOutputIs{\cdot\neq n}})$ is a counterexample for $(\mathbb{D},0)\vDash\Always\TheOutputIs{{\cdot\neq n}}$, we conclude that $(\mathbb{D} ,0)\vDash\Eventually\TheOutputIs{{\cdot= n}}$.

Consider now an attack $\alpha$, coinductively defined for $\sigma \in \ThePowersetOf{\TheSetOfOutputs}^{\TheSetOfInputs^{*}}$, %$\TheInput\in\TheSetOfInputs$, and $\TheSequenceOfInputs\in \TheSetOfInputs^{*}$, 
by $\theta_\alpha(\TheBehaviour)\triangleq\theta_\Omega(\TheBehaviour),$ and  $\delta_\alpha(\TheBehaviour)(\bullet)\triangleq\alpha(\delta^{*}_\Omega(\TheBehaviour)(\bullet,\bullet)).$ The attack $\alpha$ changes the behaviour of the dial system $\mathbb{D}$, yielding the $\TheFunctor$-coalgebra $(\AsSet{10}, \theta|_{\alpha}, \delta|_{\alpha})$, defined by 
$\theta|_{\alpha}(n)\triangleq\theta(n),$ and $
\delta|_{\alpha}(n)(\bullet)\triangleq\delta^{*}(n)(\bullet,\bullet).$ Since the attack $\alpha$ forces $\delta(\TheElement)(\bullet)$ to be $\TheElement+2~(mod~10)$ instead of $\TheElement+1~(mod~10)$, state 0 now fails to satisfy $\Eventually\TheOutputIs{\cdot =1} $, because the observation of all states that are reachable from 0 are even numbers.
\end{example}

\section{Attacker Quantification}
\label{sec:ProblemStatement}
We present the problem of attacker quantification in terms of the final $\TheFunctor$-coalgebra without losing generality thanks to the following Proposition~\ref{prop:LogicAndFinality}.

\begin{proposition}%[Coalgebraic Logic and Finality]
\label{prop:LogicAndFinality}
For an $\TheFunctor$-coalgebra $\TheCoalgebra=(\TheSet, \theta,\delta)$, $\TheElement\in \TheSet$, and for $\TheFormula\in \TheSetOfClosedAndGuardedFormulae$, 
\begin{align}
(\TheCoalgebra, \TheElement)\vDash \TheFormula \text{ if and only if } (\Omega, \TheBehaviourOf{\TheElement})\vDash \TheFormula.
\end{align}
\end{proposition}
\begin{proof} See Appendix~\ref{app:ProofLogicAndFinality}.\end{proof}

Our formulation of the attacker quantification problem is the following: given an $\TheFunctor$-coalgebra $\TheCoalgebra=(\TheSet, \theta,\delta)$, a distinguished state $\TheElement_0\in \TheSet$, a set of closed and guarded formulae $\mathbb{F}=\set{\TheFormula_1, \ldots, \TheFormula_n}$ %\subseteq \TheSetOfClosedAndGuardedFormulae$
 such that $(\TheCoalgebra,\TheElement_0) \vDash \TheFormula_1\land \ldots \land \TheFormula_n$%(i.e., $\TheElement_0$ satisfies each formula in $\mathbb{F}$)
 , we are given a set of attackers $\mathbb{A}=\set{A_1, \ldots, A_k}$, where each attacker $A\in \mathbb{A}$ is a set of attacks (i.e., $A$ is a subset of $ \ThePowersetOf{\TheSetOfOutputs}^{\TheSetOfInputs^{*}}\rightarrow \ThePowersetOf{\TheSetOfOutputs}^{\TheSetOfInputs^{*}}$). We define \emph{capabilities of an attacker $A$ at $\TheElement_0$ in the context of $\TheCoalgebra$} as the subset of formulae in $\mathbb{F}$ that do not hold in at least one of the compromised behaviours of $\TheElement_0$; more precisely 
\begin{align}
capabilities(A, \TheCoalgebra, \TheElement_0)\triangleq\bigcup_{\alpha \in A} \set{\TheFormula|\TheFormula\in \mathbb{F}~\text{and}~(\Omega, \alpha(\TheBehaviourOf{\TheElement_0})) \not \vDash \TheFormula }.
\end{align}
We say that an attacker $A_1$ is \emph{less capable} than an attacker $A_2$ at $\TheElement_0$ in the context of $\TheCoalgebra$, written $A_1\leq_{(\TheCoalgebra,\TheElement_0)} A_2$, if and only if $capabilities(A_1, \TheCoalgebra,\TheElement_0) \subseteq  capabilities(A_2, \TheCoalgebra,\TheElement_0)$. The main objective of this work is to find an efficient algorithm that calculates the relation $\leq_{(\TheCoalgebra,\TheElement_0)}$ for the attackers in $\mathbb{A}$.

We remark that the number of attacks available to an attacker and the effective capabilities of that attacker are not necessarily proportional: an attacker $A_1$ who has several attacks might not be as capable as an attacker $A_2$ that has only one attack, because the number of properties violated by $A_1$ may be smaller than the number of properties violated by $A_2$. For example, an attacker has $n$ attacks that allow her to spy in total $log~n$ of a secret of size $n$ can be considered weaker than an attacker that has only one attack which reveals the whole secret.

We use the following running example to illustrate the concepts presented the remainder of this section. 
\begin{example}[A Combination Lock]
\label{ex:SimpleLock}

By putting four dials from Example~\ref{ex:Dial} together, we form a \emph{combination lock}. 
Let $\TheSetOfInputs=\AsSet{4}=\set{0,1,2,3}$ and $\TheSetOfOutputs=\AsSet{10^4}$; we define the $\TheFunctor$-coalgebra $\mathbb{C}\triangleq(\AsSet{10^4}, \Theta,\Delta)$, where an input $\TheInput\in\AsSet{4}$ describes which dial should increase its counter. The functions $\Theta$ and $\Delta$ are defined, for $\vect{\TheElement}\in \AsSet{10^4}$ and $\TheInput, j\in\AsSet{4}$, by
\begin{align}
\Theta(\vect{\TheElement})\triangleq \set{\vect{\TheElement}}, \quad\text{ and }\quad \Delta(\vect{\TheElement})(\TheInput)(j)\triangleq \begin{cases}\vect{\TheElement}(j),\quad &\text{if $\TheInput\neq j$};\\
\delta(\vect{\TheElement}(\TheInput),\bullet),\quad &\text{otherwise}.
\end{cases}
\end{align}
Now, given a value $\vect{n} \in \AsSet{10^4}$, consider the property $\Eventually \TheOutputIs{\cdot = \vect{n}}$, which informally means ``there is a sequence of inputs that will cause the combination lock to output $\vect{n}$.'' 

For a naive model checking algorithm, to prove that the state $0000$ satisfies all properties $\Eventually\TheOutputIs{\cdot = \vect{n}}$ for $\vect{n}$ with $0000 \leq \vect{n}\leq 9999$, it would be necessary to infer 10000 counterexamples (one for each corresponding safety property). However, in the following, we present an algorithm that can infer these counterexamples in a smarter way by means of {enhanced coinduction} \cite{EnhancedCoalgebraicBisimulation,EnhancedCoinduction}.%, and we illustrate how the algorithm works through this running example.

\end{example} 
%\subsection{The $\Verify$ Algorithm}
Our model checking algorithm $\Verify$ is presented in Algorithm~\ref{alg:Verification}. $\Verify$ uses  \emph{enhanced coinduction}: a technique that builds relations whose closures given a set of algebraic operators are (bi)simulations. These relations help to reduce state exploration and to infer counterexamples.

$\Verify$ receives as input {a tuple $(\TheCoalgebra,\TheElement_0,\TheSafetyFormula_0, R_0, F_0)$ of an $\TheFunctor$-coalgebra $\TheCoalgebra=(\TheSet, \theta,\delta)$, a state $\TheElement_0$, a safety formula $\TheSafetyFormula_0$, a
relation $R_0$ where $\TheElement~R_0~\TheSafetyFormula\Rightarrow\TheElement\vDash\TheSafetyFormula$; and a 
	%hash-map 
	relation $F_0$ where $\TheElement~F_0~\TheSafetyFormula\Rightarrow\TheElement\not\vDash\TheSafetyFormula$. %and values in $\TheSetOfInputs^{*}\cup\set{\perp}$, and a set of pairs $(\beta,\gamma)$
}
After execution, $\Verify$ returns {a tuple $(res,R,F)$ where $res$ equals $true$ if $(\TheCoalgebra,\TheElement_0)\vDash\TheSafetyFormula_0$, otherwise $res$ equals $false$; 
 %is a counterexample $(\TheElement_{\TheCounterexample},\TheSafetyFormula_{\TheCounterexample})$; 
 	a relation $R$, where, if $(\TheElement,\TheSafetyFormula)\in R$, then $\TheElement\vDash\TheSafetyFormula$; and a 
	%hash-map 
	relation $F$ where, if $(\TheElement,\TheSafetyFormula)\in F$, then $\TheElement\not\vDash\TheSafetyFormula$.}\,
\begin{algorithm}[t]
{\footnotesize
{$R\leftarrow R_0$}\;
 $F\leftarrow F_0$\;
%{$todo\leftarrow\left[(\varepsilon,\TheElement_0,\TheSafetyFormula_0)\right]$}\;
{$todo\leftarrow\left[(\TheElement_0,\TheSafetyFormula_0)\right]$}\;
\While{$todo\neq [ ]$}{
% 	$(c_{ex},\TheElement, \TheSafetyFormula)\leftarrow todo.\AsFunction{pop}()$\;
 	$(\TheElement, \TheSafetyFormula)\leftarrow todo.\AsFunction{pop}()$\;
	%  $c'_{ex}\leftarrow $\;
	 %\If{$(\TheElement, \TheSafetyFormula)\in\invcrtu(F)$}{
	 \If{$\crtu(\set{(\TheElement, \TheSafetyFormula)})\cap F \neq \emptyset$}{
	 \label{line:inF}
	 	 %/\!/Infer failure\;
  	 	%$F.\AsFunction{insert}(\TheElement, \TheSafetyFormula)$\; //no need to add the pair. can already infer it
		 \Return $(false, \TheRelation_0,F)$\;%\Return $((\TheElement, \TheSafetyFormula), \TheRelation_0,F)$\;
	 }
  	\If{$(\TheElement, \TheSafetyFormula) \not\in\crtu(R)$}{
 	 \label{line:notR}
   	% /\!/Could not infer satisfiability {\color{white}\;}
		\eIf{$\TheOutputOf{\TheElement} \not\subseteq \theta(\TheSafetyFormula)$}{\label{line:notF}{
				$F.\AsFunction{insert}(\TheElement,\TheSafetyFormula)$\;
			%}
			\Return $(false,R_0, F)$\;%\Return $((\TheElement,\TheSafetyFormula),R_0, F)$\;%//Could return a counterexample instead\; 
		}
  	 }
  	 {%else
   	$R.\AsFunction{insert}(\TheElement, \TheSafetyFormula)$\; 
   	\For{$\TheInput \in \TheSetOfInputs$}{
		\If {$\crtu(\set{(\delta(\TheElement,\TheInput), \delta(\TheSafetyFormula,\TheInput))})\cap F \neq \emptyset$}{
		%/\!/ Found a path towards a counterexample\,
		\label{line:nextBad}
		$todo.\AsFunction{prepend}((c_{ex}:\TheInput),\delta(\TheElement)(\TheInput), \delta(\TheSafetyFormula)(\TheInput))$\;
			 \Break \;
		}	 
		\If{$(\delta(\TheElement)(\TheInput), \delta(\TheSafetyFormula)(\TheInput))\not\in \crtu(R)$}{\label{line:nextOk}
			%$todo.\AsFunction{append}((c_{ex}:\TheInput),\delta(\TheElement)(\TheInput), \delta(\TheSafetyFormula)(\TheInput))$\;
			$todo.\AsFunction{append}(\delta(\TheElement)(\TheInput), \delta(\TheSafetyFormula)(\TheInput))$\;
		}
	}
	}
   	}
}
 \Return $(true,R, F_0)$\;\,
 \caption{The $\Verify$ algorithm. 
 }}
 \label{alg:Verification}
\end{algorithm}

$\Verify$ relies on the same principles as the $\AsFunction{HKC}$ algorithm \cite{Bonchi:2013:CNE:2480359.2429124} (a version of Hopcroft and Karp's algorithm \cite{HopcroftAndKarp} that uses enhanced coinduction to decide language equivalence of finite automata). Our use of enhanced coinduction is marked by the appearance of the \emph{precongruence closure} function $\crtu\colon \ThePowersetOf{\TheSet\times\TheSetOfClosedAndGuardedSafetyFormulae}\rightarrow\ThePowersetOf{\TheSet\times\TheSetOfClosedAndGuardedSafetyFormulae}$ in Lines~\ref{line:inF},~\ref{line:notR}%,~\ref{line:notF}
,~\ref{line:nextBad}, and~\ref{line:nextOk}. 

Using precongruence closures, we can build \emph{simulation up-to-precongruences} that speed up the verification process by helping in the inference of satisfaction/failure of the formulae we are checking. To define these closures, it is desirable to have a set of \emph{algebraic operators}. In a nutshell, algebraic operators are functions that preserve the satisfaction relation between states and safety formulae.

\begin{definition}[Algebraic Operator]
Let $\TheCoalgebra=(\TheSet,\theta,\delta)$ be an $\TheFunctor$-coalgebra; we say that a function $\beta\colon \TheSet\times \TheSetOfClosedAndGuardedSafetyFormulae\rightarrow\TheSet\times \TheSetOfClosedAndGuardedSafetyFormulae$ is an \emph{algebraic operator} (in the context of $\TheCoalgebra$) if and only if $\beta$ is monotone with respect to $\vDash$. 
 \end{definition}
Besides the obvious identity operator, the following operators exist for any $\TheFunctor$-coalgebra. Let $\TheCoalgebra=(\TheSet,\theta,\delta)$ be an $\TheFunctor$-coalgebra. For all $\TheInput\in \TheSetOfInputs$, we the operator $\xrightarrow{i}\colon \TheSet\times\TheSetOfClosedAndGuardedSafetyFormulae\rightarrow\TheSet\times\TheSetOfClosedAndGuardedSafetyFormulae$, for $\TheElement\in \TheSet$ and $\TheSafetyFormula\in \TheSetOfClosedAndGuardedSafetyFormulae$, is defined by 
\begin{align}
\xrightarrow{i}(\TheElement,\TheSafetyFormula)\triangleq(\delta(\TheElement)(\TheInput),\delta(\TheSafetyFormula)(\TheInput)\land\TheOutputIs{\theta(\delta(\TheElement)(\TheInput))}).
\end{align}
This construction ``forces'' the behaviour of $\TheElement$ into the formula $\TheSafetyFormula$. If $\delta(\TheSafetyFormula)(\TheInput)\land\TheOutputIs{\theta(\delta(\TheElement)(\TheInput))}$ is $\False$, then $(\delta(\TheElement)(\TheInput),\delta(\TheSafetyFormula)(\TheInput))$ must be a counterexample for $(\TheElement,\TheSafetyFormula)$, since $\theta(\delta(\TheElement)(\TheInput))$ cannot be a subset of $\theta(\delta(\TheSafetyFormula)(\TheInput))$. This condition bears a remarkable similarity to the acceptance conditions used in the explicit state model checking, which uses product automata with a special acceptance condition \cite{HandbookOfModelChecking}.

We can finally introduce the precongruence closure operator $\crtu$, which enhances contextual closures with transitivity-and-union (of simulation relations).
\begin{definition}[Precongruence Closure]
\label{def:crtu}
Let $(\TheSet,\theta,\delta)$ be an $\TheFunctor$-coalgebra, let $\leq\ \subseteq \TheSet\times\TheSet$ be a simulation relation in $\TheSet$ and let $\leadsto\ \subseteq\TheSetOfClosedAndGuardedSafetyFormulae\times\TheSetOfClosedAndGuardedSafetyFormulae$ be a simulation in $\TheSetOfClosedAndGuardedSafetyFormulae$.
 Given a relation $\TheRelation \subseteq \TheSet\times\TheSetOfClosedAndGuardedSafetyFormulae$, we inductively define the \emph{precongruence closure} $\crtu(R)$ of $R$ using the following inference rules:
\begin{align*}
{\inference[id: ]{\TheElement~R~\TheFormula}{\TheElement~\crtu(R)~\TheFormula}},%\quad{\inference[r: ]{true}{a~\crtu(R)~a}},
\quad\quad\quad\quad{\inference[$c_{\beta}$: ]{\TheElement~R~\TheFormula,~\beta(\TheElement,\TheFormula)=(\TheElement',\TheFormula')}{\TheElement'~\crtu(R)~\TheFormula'}},\quad
%{\inference[t: ]{a~\crtu(R)~b, b~\crtu(R)~c}{a~\crtu(R)~c}},
\end{align*}
\begin{align*}
{\inference[$tu_\leq$: ]{x\leq y,~y~\crtu(R)~\TheFormula}{x~\crtu(R)~\TheFormula}},\quad
{\inference[$tu_\Rightarrow$: ]{\TheElement~\crtu(R)~\TheFormula,~\TheFormula~\leadsto~\TheFormula'}{\TheElement~\crtu(R)~\TheFormula'}},\end{align*}
(Since we are only interested in element-formula pairs for model checking, we omit including reflexive and symmetry rules in the precongruence closure.)

If $\crtu(R)$ is a simulation between $\TheSet$ and $\TheSetOfClosedAndGuardedSafetyFormulae$, then we say that $\TheRelation$ is a \emph{simulation up-to-precongruence.}
\end{definition}
\begin{proposition}[Satisfaction Through Enhancement]
\label{prop:crtu}
Let $\TheCoalgebra=(\TheSet,\theta,\delta)$ be an $\TheFunctor$-coalgebra, $\TheElement\in\TheSet$ be a state, and $\TheFormula \in \TheSetOfClosedAndGuardedSafetyFormulae$ be a safety formula; if $(\TheCoalgebra,\TheElement)\vDash\TheFormula$, then $(\TheCoalgebra,\TheElement')\vDash\TheFormula'$ for all $(\TheElement',\TheFormula')\in\crtu\left(\set{(\TheElement,\TheFormula)}\right)$.
\end{proposition}
\begin{proof}
We prove that $(\TheCoalgebra,\TheElement')\vDash\TheFormula'$ for the different cases of $(\TheElement',\TheFormula')$:
\begin{itemize}
\item $id$: if $(\TheElement',\TheFormula')=(\TheElement,\TheFormula)$, then $(\TheCoalgebra,\TheElement')\vDash\TheFormula'$ trivially holds since $(\TheCoalgebra,\TheElement)\vDash\TheFormula$
\item $c_\beta$: if $(\TheElement',\TheFormula')=\beta(\TheElement,\TheFormula)$, then $(\TheCoalgebra,\TheElement')\vDash\TheFormula'$ holds because $\beta$ is monotone with respect to $\vDash$.
\item $tu$: if $\TheElement'~S~\TheElement$ and $\TheFormula\Rightarrow \TheFormula'$, then $(\TheCoalgebra,\TheElement')\vDash\TheFormula'$ since the composition of simulation relations is a simulation relation, and all simulation relations from $\TheCoalgebra$ to $\zeta$ are subsets of $\vDash$.
\item $ctu$: if $\TheElement'~S~\beta(\TheElement)$ and $\gamma(\TheFormula)\Rightarrow \TheFormula'$, then $(\TheCoalgebra,\TheElement')\vDash\TheFormula'$ follows from a similar argument.\qed
\end{itemize}
\end{proof}

\begin{corollary}[Counterexample Through Enhancement]
\label{cor:InferCounterexamples}
If there exists a pair $(\TheElement', \TheFormula')$ in $\crtu\left(\set{(\TheElement,\TheFormula)}\right)$ such that $(\TheCoalgebra,\TheElement')\not\vDash\TheFormula'$, then $(\TheCoalgebra,\TheElement)\not\vDash\TheFormula$. Hence, $(\TheElement', \TheFormula')$ is a counterexample for $(\TheCoalgebra,\TheElement)\vDash\TheFormula$.
\end{corollary}
\begin{proof}
By a contrapositive argument of Proposition~\ref{prop:crtu}.\qed
\end{proof}
We use Corollary~\ref{cor:InferCounterexamples} in Line~\ref{line:inF} of the $\Verify$ algorithm to infer violation of properties, considering that $F$, for each pair $(\TheElement,\TheFormula)$ in $F$, we have $(\TheCoalgebra,\TheElement)\not\vDash\TheFormula$.

The $\Verify$ algorithm heavily relies on Proposition~\ref{prop:crtu} and Corollary~\ref{cor:InferCounterexamples} for its correctness: at its core, $\Verify$ builds a relation $R$, and treats it as if it was a simulation up-to precongruence, until we either find a counterexample or we finish, and we can conclude that all elements in $R$ are valid state-formula pairs. 

$\Verify$ explores a space of size $|\TheElement_0||\TheSafetyFormula_0|$, where $|\TheElement_0|$ is the number of states that are reachable from $\TheElement_0$. For each pair $(\TheElement, \TheSafetyFormula)$ that has not been visited, the algorithm has to check whether $\theta(\TheElement)\subseteq \theta(\TheSafetyFormula)$, which is lineal in the size of $\TheSetOfOutputs$. Without any optimisation, $\Verify$ would have a worst time complexity of $\mathcal{O}(|O||\TheSet||\TheSafetyFormula_0|)$. Moreover, if we have $n$ formulae to check, the worst case complexity of the attacker quantification problem is $\mathcal{O}(n|O||\TheSet||\TheSafetyFormula_{max}|)$ where $|\TheSafetyFormula_{max}|$ is the biggest formula. Our intention is to reduce the complexity in two fronts: one, to change the worst case complexity to $\mathcal{O}(a|O||\TheSet||\TheSafetyFormula_{max}|+ b|R| +c|F|)$ by allowing the inference of, ideally with $\mathcal{O}(a)<\mathcal{O}(n)$, and second, reduce the complexity of a single execution of $\Verify$ from $\mathcal{O}(|O||\TheSet||\TheSafetyFormula_0|)$ to $\mathcal{O}(k|O||\TheSet|+j|R|)$, ideally with $\mathcal{O}(k)<\mathcal{O}(|\TheSafetyFormula_0|)$. We remark that the use of algebraic operators does not guarantee a reduction in the overall complexity of the algorithm, and may even cause unnecessary overhead; however, our experiments show that we can achieve reductions in the size of the state space by using algebraic operators.

We now recall the research questions that motivate this work: \textbf{(RQ1)} how can we reduce the effort of verifying multiple properties for the quantification of attacker capabilities? and \textbf{(RQ2)} how do we quantify the attackers using the verification results?

During the remainder of this section, we answer RQ1 by performing experiments in the combination lock and a BEEM (Benchmarks for Explicit Model Checkers) problem \cite{BeemDatabase}. More precisely, we use different sets of algebraic operators to reduce the number of states that the $\Verify$ algorithm has to explore or reduce the number of properties hat have to be verified. We show an answer for RQ2 in Section \ref{sec:CaseStudy}.

%, without introducing too much overhead. 
Let us consider two algebraic operators for the combination lock $\mathbb{C}$: an operator $\shift$ that shifts the components of the dial to the right, and an operator $\addIt$ that adds 1 to the rightmost component of the dial. Formally, we define $\shift:\AsSet{10^4}\times\TheSetOfClosedAndGuardedSafetyFormulae\rightarrow \AsSet{10^4}\times\TheSetOfClosedAndGuardedSafetyFormulae$, for $(\TheElement_0,\TheElement_1,\TheElement_2,\TheElement_3)\in\AsSet{10^4}$ and $\TheFormula_1, \TheFormula_2\in\TheSetOfClosedAndGuardedSafetyFormulae$, in a component wise manner (i.e., $\shift=(\shift_1, \shift_2)$) by 
\begin{align*}
\shift_1(\TheElement_0,\TheElement_1,\TheElement_2,\TheElement_3)&\triangleq(\TheElement_3,\TheElement_0,\TheElement_1,\TheElement_2);
\end{align*}
\vspace{-0.86cm}
\begin{align*}
\shift_2(\TheOutputIs{\OtherPredicate})&\triangleq\TheOutputIs{\shift_1(\OtherPredicate)},&\shift_2(\TheFormula_1 \land \TheFormula_2)&\triangleq\shift_2(\TheFormula_1)\land \shift_2(\TheFormula_2)\\
\shift_2(\After{\ThePredicate}\TheFormula)&\triangleq\After{\ThePredicate}\left(\shift_2(\TheFormula)\right),&
\shift_2(\nu\TheVariable.\TheGuardedFormula)&\triangleq\shift_2(\TheGuardedFormula\left[\nu\TheVariable.\TheGuardedFormula/\TheVariable\right]),\\
\end{align*}
where $\shift_1(Q)\triangleq\set{\shift_1(\vect{q})|\vect{q}\in Q}$. Similarly, we define $\addIt:\AsSet{10^4}\times\TheSetOfClosedAndGuardedSafetyFormulae\rightarrow \AsSet{10^4}\times\TheSetOfClosedAndGuardedSafetyFormulae$, for $(\TheElement_0,\TheElement_1,\TheElement_2,\TheElement_3)\in\AsSet{10^4}$ and $\TheFormula_1, \TheFormula_2\in\TheSetOfClosedAndGuardedSafetyFormulae$, by 
\begin{align*}
\addIt_1(\TheElement_0,\TheElement_1,\TheElement_2,\TheElement_3)&\triangleq(\TheElement_0,\TheElement_1,\TheElement_2,\TheElement_3+1~(mod~10));
\end{align*}
\vspace{-0.86cm}
\begin{align*}
\addIt_2(\TheOutputIs{\OtherPredicate})&\triangleq\TheOutputIs{\addIt_1(\OtherPredicate)},&\addIt_2(\TheFormula_1 \land \TheFormula_2)&\triangleq\addIt_2(\TheFormula_1)\land \addIt_2(\TheFormula_2)\\
\addIt_2(\After{\ThePredicate}\TheFormula)&\triangleq\After{\ThePredicate}\left(\addIt_2(\TheFormula)\right),&
\addIt_2(\nu\TheVariable.\TheGuardedFormula)&\triangleq\addIt_2(\TheGuardedFormula\left[\nu\TheVariable.\TheGuardedFormula/\TheVariable\right]),
\vspace{-0.86cm}
\end{align*}
where $\addIt_1(Q)\triangleq\set{\addIt_1(\vect{q})|\vect{q}\in Q}$. We also define $\shift^2$, $\shift^3$, and $\addIt^2$ to $\addIt^9$ by composing the operators with themselves. 

To provide an intuition of how formulae can be inferred from previous verification experiences, let us describe the first steps of a verification procedure that uses the operators $\shift$ and $\addIt$. We start with $R=\set{}$ and $F=\set{}$, then, for $(0000,\Always\TheOutputIs{\cdot \neq {0000}})$, $\Verify$ concludes that $(0000,\Always\TheOutputIs{\cdot \neq {0000}})$ is a counterexample, and it returns $(false,\set{},\set{(0000,\Always\TheOutputIs{\cdot \neq {0000}})})$. We can henceforth infer that any pair $(\TheElement, \TheFormula)$ such that $(0000,\Always\TheOutputIs{\cdot \neq {0000}})\in \crtu(\set{(\TheElement, \TheFormula)})$ is a counterexample for the formula that we are currently checking. Note that, since we only considered the $\shift$ and $\addIt$ operators, this condition holds only for $(0009,\Always\TheOutputIs{\cdot \neq {0009}})$; nevertheless, we could extend this condition to all pairs of the form $(000x,\Always\TheOutputIs{\cdot \neq {000x}})$ if we considered more $\addIt^i$ operators.

We continue verifying with the knowledge we have so far: with $R=\set{}$ and $F=\set{(0000,\Always\TheOutputIs{\cdot \neq {0000}})}$, for $(0000,\Always\TheOutputIs{\cdot \neq {0001}})$, $\Verify$ returns
 \[(false,\set{},\set{(0000,\Always\TheOutputIs{\cdot \neq {0000}}),(0000,\Always\TheOutputIs{\cdot \neq {0001}})})\]
Henceforth, for any $(\TheElement, \TheFormula)$ with $(0000,\Always\TheOutputIs{\cdot \neq {0001}})\in \crtu(\set{(\TheElement, \TheFormula)})$, we can infer that $(\TheElement, \TheFormula)$ is a counterexample for the formula we are currently checking. For example, the verification of $\Always\TheOutputIs{\cdot \neq {0010}}$ becomes ``trivial,'' since $\Verify$ infers its failure at the initial state. 

We observe that use of algebraic operators during verification can make some properties trivial to verify, since their validity can be inferred at the initial state. Table~\ref{tab:CombinationLock} shows the number of properties that become trivial given different combinations of operators. We see that the more operators we consider, the more likely it is for $\Verify$ to infer properties.
\begin{table}[t]
\centering
\begin{tabular}{|c|c|| c|c|}
\hline
Operators& Inferred &Operators& Inferred \\
\hline
 $\set{\shift}$  & 3675 & $\set{\shift,\shift^2,\shift^3}$& 7470\\
 \hline
$\set{\addIt}$& 5000 & $\set{\shift,\shift^2,\shift^3,\addIt}$& 7470\\
\hline
 $\set{\shift,\addIt}$  & 5925 &$\set{\shift,\shift^2,\shift^3,\addIt,\addIt^2}$ &7550\\
 \hline
 $\set{\shift,\shift^2}$  & 4995 & $\set{\shift,\shift^2,\shift^3,\addIt,\addIt^2,...,\addIt^9}$& 9046\\
 \hline
\end{tabular}
\caption{Inferred properties given a set of algebraic operators (out of 10000) in the combination lock example.}
\label{tab:CombinationLock}
\vspace{-0.5cm}
\end{table}

In the bounded concurrent adding puzzle (from \cite{BeemDatabase}), there are two processes \texttt{P} and  \texttt{Q} running in parallel, and they write over shared accumulator $c$, with \texttt{P =
    loop \{x=c; x=x+c; c=x;\}} and \texttt{Q = loop \{y=c; y=y+c; c=y;\}}. The initial value of $c$ is $1$, and the claim is that $c$ can eventually take any given natural value. With $\TheSetOfInputs = \set{1,2}$, and $\TheSetOfOutputs=\mathbb{N}$, and $\TheSet=\set{Q,R,S}\times\mathbb{N}$, we model the addition puzzle problem $\TheCoalgebra=(\TheSet\times\TheSet\times\mathbb{N}, \Theta,\Delta)$ with an interleaving product of two coalgebras defined by
\begin{align*}
\theta(\TheElement,n)&\triangleq n,\\
\delta(\TheElement,n)(c)&\triangleq 
\begin{cases}
(R,c) &\text{if $\TheElement=Q$ and $c<MAX$}\\
(S,n+c), &\text{if $\TheElement=R$}\\
(Q,n), &\text{if $\TheElement=S$}\\
(\TheElement,n), &\text{otherwise}.
\end{cases}\\
\Theta((\TheElement_1,n_1),(\TheElement_2,n_2),c)&\triangleq \set{c},\\
\Delta((\TheElement_1,n_1),(\TheElement_2,n_2),c)(\TheInput)&\triangleq 
\begin{cases}
\Delta(\delta(\TheElement_1,n_1)(c),(\TheElement_2,n_2),c),\quad &\text{if $\TheInput=1$ and $\TheElement_1\neq S$};\\
\Delta(\delta(\TheElement_1,n_1)(c),(\TheElement_2,n_2),n_1),\quad &\text{if $\TheInput=1$ and $\TheElement_1= S$};\\
\Delta((\TheElement_1,n_1),\delta(\TheElement_2,n_2)(c),c),\quad &\text{if $\TheInput=2$ and $\TheElement_2\neq S$};\\
\Delta((\TheElement_1,n_1),\delta(\TheElement_2,n_2)(c),n_2),\quad &\text{if $\TheInput=2$ and $\TheElement_2 = S$};
\end{cases}
\end{align*}
For $c$ with $1\leq c \leq MAX$, we prove that the initial state, $((Q,0),(Q,0),1)$, satisfies $\Eventually \TheOutputIs{\cdot = c}$, which is equivalent to solving the puzzle for $c$.

We now define the $\mathtt{swap}$ operator that swaps both processes by 
\begin{align}
\mathtt{swap}(((\TheElement_1,n_1),(\TheElement_2,n_2),c),\TheFormula) &\triangleq ((\TheElement_2,n_2), (\TheElement_1,n_1),c), \TheFormula).
\end{align}
Table \ref{tab:Addition} shows the impact of using $\mathtt{swap}$ during verification. We observe that the $\mathtt{swap}$ operator consistently reduces state exploration by almost half. 
\begin{table}[t]
\centering
\begin{tabular}{|c|c|c||r|r|r|r|}
\hline
%$(N,Max)$ &   \multicolumn{2}{c|}{None}&  \multicolumn{2}{c|}{$\mathtt{swap}$}&  \multicolumn{2}{c|}{$\mathtt{nxt_i}$} &  \multicolumn{2}{c|}{Both}\\
$(N,MAX)$ &  {None}&   {$\mathtt{swap}$}&  {None}&  {$\mathtt{swap}$} \\
\hline
(17, 20)	& 1616 &  845    &0.76 s &0.40 s\\
\hline
(500, 200)	& 14298 &  7937   &6.63 s   & 	3.82 s \\
\hline
(637, 300)	&485942 &  247602 & 252.01 s &145.84 s\\
\hline
(749, 400)	&845020 &   425093  & 446.25 s &267.58 s \\
\hline
\end{tabular}
\caption{States explored and average time for the concurrent addition puzzle.}
\label{tab:Addition}
\vspace{-0.5cm}
\end{table}

\section{A Case Study}
\label{sec:CaseStudy}
To illustrate the problem of attacker quantification through model checking, we provide a small case study based on the water treatment testbed SWaT \cite{SWaT}. SWaT is an operational test bed, built for experimental research in the design and security of industrial control systems ICS in the Singapore University of Technology and Design. SWaT is miniature version of the actual water treatment plants in Singapore, capable of producing 5 gallons/minute of treated water. The SWaT testbed consists of six physical processes controlled by Programming Logic Controllers (PLCs), a human-machine interface, a SCADA system, and a historian. Each physical process has its own sensors and actuators, including pressure meters, flow meters and valves among others.

For this example, we consider a part of process 1 of the water treatment cycle, which is in charge of water collection. We show how we can translate a relationship between physical properties of the system into something that can be used by $\Verify$ to check for properties.

Consider the composition of a controller $C$, a water tank $t101$, a water level sensor $lit101$, a water pressure sensor $hg101$, and an input valve $mv101$. In this process, the water level and water pressure of the system are related by \emph{hydrostatic pressure}: a physical invariant modelled by the equation $hg101\approx \mathtt{g}*lit101$, where $\mathtt{g}$ is the acceleration caused by gravity.

We model the valve $mv101$ as two-state systems whose state represents whether the valve is open or closed. The input to the valve is $c_v\in \set{open,close}$, and sets its next state. There are different cases: if $lit101 <500$, then $c_v=open$; if $lit101 >800$, then $c_v=close$, and $c_v=mv101$ otherwise.

We model the water tank as a system whose state is defined by its current water level, and is bounded by a capacity of $1200$. The input to the tank is its current inflow  $in$, and its state transitions are defined by the inflow-outflow equation $t'=t+in-out$ with $out = 0.44$ (enforcing $0\leq t' \leq 1200$), where $t$ is the current water level of the tank. The input $in$ in turn respectively depends on the state of the valve $mv101$ as follows: {if $mv101=open$, then $in=0.46$; otherwise, $in=0$.

We can model this process as a $\TheFunctor$-coalgebra with $\TheSetOfInputs =\set{\bullet}$ and $\TheSetOfOutputs = L\times P\times \mathbb{B}$, where $L$ is the set of valid water level readings and $P$ is the set of valid water pressure readings. We define the $\TheFunctor$-coalgebra $\mathbb{P}=(L\times L\times P, \theta,\delta)$; a state is a tuple $(t, lit101,hg101)$ where $t$ is the real water level of the tank and $a$ is the state of the alarm in the controller. The functions $\theta$ and $\delta$ are defined by
\begin{align}
\theta(t, lit101,hg101)&\triangleq(t,\mathtt{g}*t,hg101\approx \mathtt{g}*lit101),\\ 
\delta(t,lit101,hg101)&\triangleq(t',t,\mathtt{g}*t)
\end{align}
We model the hydrostatic pressure enforced by physics with the formula
\[Hydro\triangleq \Always (\TheOutputIs{(t,p,a)\approx (t,\mathtt{g}*t,a)}),\]
while the relationship $hg101\approx \mathtt{g}*lit101$ is checked as part of the observation of the state. Whenever this relationship between sensor readings does not hold, we raise an alarm. Now, consider the following safety requirements:
\begin{enumerate}
\item The real water level is always within safe levels,
\item The real water pressure is always within safe levels,
\item Consistent readings between the water level sensor and the pressure sensor,
\end{enumerate}
These three properties can be respectively modelled by the following formulae,
\begin{align}
Lvl&\triangleq\Always (\TheOutputIs{(t,\_~,\_~)\geq 200} \land\TheOutputIs{(t,\_~,\_~)\leq 1000} ),\\
Hg&\triangleq\Always (\TheOutputIs{(\_~,p,\_~)\geq 1000} \land\TheOutputIs{(\_~,p,\_~)\leq 9000} ),\\
Con&\triangleq\Always (\TheOutputIs{(\_~,\_~,a)=True}).
\end{align}
%These properties are interesting because they allow us to detect attackers that controls one of the sensors, either the pressure or the water level sensor, by relying on the physical invariant that relates the two, as shown by \cite{Adepu:2016:DDS:2897845.2897855}.
Since these safety formulae have coalgebraic semantics, we can compute the largest similarity relation in the set $\set{Hydro, Lvl, Hg, Con}$ prior to verification, which yields the simulation relation $\set{(Hydro \land Lvl,P)| P \in \set{Lvl, Hg} }$; from this relation, we can infer $Hg$ from $Hydro \land Lvl$ should $Lvl$ be satisfied ($Hydro$ is itself enforced by the laws of physics, so there is no need to check it). Thus, we can sort our proof obligations as $[Lvl, Hg, Con]$, so that, if $Lvl$ holds, then we can infer $Hg$ from the precongruence closure using the rule $tu_\Rightarrow$. 

%\subsection{Attackers}
Consider the following attacks over sensor readings, which are modelled after realistic attacks for CPSs \cite{Cardenas:2011:AAP:1966913.1966959}; similar attacks have been carried out in the real testbed during the SWaT Security Showdown \cite{S3Showdown}:
\begin{enumerate}
\item Set $lit101$ to its maximum value (surge attack),
\item Add a bias $b$ to the current value of $lit101$ (bias attack),
\item Add a bias $b$ to the current value of $lit101$, and produce a consistent reading from $hg101$ (stealthy attack).
\end{enumerate}
We model these attacks respectively with the following functions:
\begin{align}
SurgeUp(t,lit101,hg101)&\triangleq(t,1200,hg101)\\
Bias_b(t,lit101,hg101)&\triangleq(t,lit101+b,hg101)\\
Stealthy_b(t,lit101,hg101)&\triangleq(t,lit101+b, hg101+\mathtt{g}*b)
\end{align}
We use these functions to define new $\TheFunctor$-coalgebras that capture the effect of these attacks. Let $\delta_\alpha\triangleq\delta\circ SurgeUp$, $\delta_\beta\triangleq\delta\circ Bias_{200} $ and $\delta_\gamma\triangleq\delta\circ Stealthy_{500}$
; we preserve the observation function $\theta$ as it is. We will consider three attackers: $\set{\alpha}$, $\set{\beta}$, and $\set{\gamma}$.
%\subsection{Results}

Table \ref{tab:Swat} illustrates how many states need exploration when verifying the three properties. We make measurements for the original system and for the system under attacks $\alpha$, $\beta$ and $\gamma$. For this exercise, our initial state is $(500,\mathtt{g}*500,True)$, and we reuse the knowledge of previous verifications (e.g., we use whatever we learn from the verification of $Lvl$ to verify $Hg$ and $Con$).
\begin{table}[t]
\centering
\begin{tabular}{|c|c||c|c|c|}
\hline
Formula  &$(\theta,\delta)$& $(\theta, \delta_\alpha) $ & $(\theta, \delta_\beta)$& $(\theta, \delta_\gamma)$\\
\hline
$Lvl$ & $15757,True$ & $683,False$ & $16199,True$ & $683,False$\\
\hline
$Hg$ & $1,True$ & $683,False$ & $1,True$ & $683,False$\\
\hline
$Con$ &15757,True& $2, False$ & $2, False$ & $1139,True$ \\
\hline
\end{tabular}
\caption{This table shows for a given formula and an attacker model how many states were explored by the verification algorithm and whether the property holds or not for state $(500,\mathtt{g}*500,True)$.}
\label{tab:Swat}
\vspace{-0.5cm}
\end{table}
According to the definition of attacker capabilities from Section \ref{sec:ProblemStatement}, obtain the following results:
\begin{align*}
capabilities(\set{\alpha},\mathbb{P},(500,\mathtt{g}*500,True))&=\set{Lvl,Hg,Con}\\
capabilities(\set{\beta},\mathbb{P},(500,\mathtt{g}*500,True))&=\set{Con}\\
capabilities(\set{\gamma},\mathbb{P},(500,\mathtt{g}*500,True))&=\set{Lvl, Hg}
\end{align*}
These results yield the following attacker hierarchy: $\set{\beta} \leq \set{\alpha}$ and $\set{\gamma}\leq \set{\alpha}$. Attackers $\set{\beta}$ and $\set{\gamma}$ are unrelated.

%\subsection{Discussion}
We might want to apply filters to obtain further interesting classifications. For example, if we are looking for attackers that do not try to avoid detectability, then we focus on attackers that have have $Con$ as part of their their capabilities (in this case, $\set{\alpha}$ and $\set{\beta}$). This result fits the intuition that attacker $\set{\gamma}$ puts effort into creating a consistent reading for both sensors, unlike attackers $\set{\alpha}$ and $\set{\beta}$.  
For another example, if we are concerned about attackers that can impact the physical aspect of the system in a noticeable way (e.g., by causing the water tank to have an unsafe water level), then we look for attackers with $Lvl$ in their set of capabilities (in this case, $\set{\alpha}$ and $\set{\gamma}$). 

\section{Related Work}
Basin and Cremers \cite{KnowYourEnemy} define a unified framework for the definition of attackers and protocols; the authors associate attackers with a set of attacker actions, and then infer the set of attacks available to the attacker from these actions. In their framework, attacks correspond to a sequence of actions that an attacker takes during a protocol execution. We adopt a different approach, and we consider attackers to be sets of functions that change the behaviour of the system (i.e, attacks); a systematic construction of attacks is outside of the scope of this work.

For CPSs, Rothstein et al. \cite{Simei} compare attackers of CPSs using a measure of control over physical properties. In \cite{Simei}, each attacker is given set of actuators and sensors under her control, and by quantifying her interference over a physical variable, they provide a measure of the power of the attacker (for example, attackers that can both empty and overflow a water tank are considered strictly more powerful than attackers that can only empty the tank). 

We can also consider risk assessment to be related to the problem of quantifying the capabilities of attackers. In particular, we would like to highlight the work of Cardenas et al. \cite{Cardenas:2011:AAP:1966913.1966959} in risk assessment for CPSs, where they claim that the most well known risk metric is probably the average loss, which depends on the probability of occurrence of an event and proportional to the estimated loss caused by the event. Intuitively, this claim is not exclusive to CPSs, and there are several ways to mathematically define average loss.

With respect to the verification of multiple properties, we would like to mention the work by Goldberg et al. \cite{JustAssume}, where they consider the problem of efficiently checking a set of safety properties $P_1$ to $P_k$. Their $JA$-verification (short for Just-Assume) consists of individually checking each property, while assuming that all other properties are valid. However, their approach does not guarantee the identification of all failed properties, which is a vital requirement for our attacker capability quantification. We thus needed to adopt an alternative approach.

To the best of our knowledge, our work is the first to use enhanced coinduction for the model checking of multiple properties. Bonchi and Pous \cite{Bonchi:2013:CNE:2480359.2429124} use bisimulation up to congruence to prove language equivalence of non-deterministic finite automata by optimising the classical algorithm by Hopcroft and Karp, and Rot provides in \cite{EnhancedCoinduction} a comprehensive study of enhanced coinduction.

\section{Conclusion}
We illustrated by means of a case study how it is possible to quantify attackers via coalgebraic model checking multiple properties. The case study is based on an existing water treatment testbed, and it shows how attackers can be classified in terms of the properties that they violate.

We show that using enhanced coinduction for the model checking of multiple properties can reduce both the number of properties to be verified and the state space of individual verifications at the linear cost of computing the precongruence closure over the relations used to store knowledge. While our experiments show that it is possible to greatly reduce the number of properties to be verified (e.g., from 10000 to only 954) given an adequate choice of algebraic operators, we remark that using enhanced coinduction does not guarantee a reduction in the effort required to verify multiple properties, since computing the precongruence closure requires an effort linear in the size of the relation that we use to store knowledge. From our experiments, we see a greater reduction in state exploration if the algebraic operators inform the model checking algorithm of any algebraic structure in the state space; for example, if there are symmetries in the state space.
\appendix
\section*{Appendix}
\section{Proof of Proposition~\ref{prop:LogicAndFinality}}
\label{app:ProofLogicAndFinality}
\begin{proof}

Since bisimulations correspond to behavioural equivalence relations in deterministic systems, and since coalgebraic logics respect bisimulations, it suffices to show that the graph of $\TheBehaviourOf{-}$ is a bisimulation relation from $\TheCoalgebra$ to $\Gamma({\TheBehaviourOf{x_0}})$ to prove the proposition.

We define the graph of $\TheBehaviourOf{-}$ by $G(\TheBehaviourOf{-})\triangleq\set{(x,\TheBehaviourOf{\TheElement})|x\in X}$. To show that $\TheBehaviourOf{-}$ is a bisimulation, we must show, for all $(x,\TheBehaviourOf{\TheElement})\in G(\TheBehaviourOf{-})$:
\begin{enumerate}
\item that $(x_0,\TheBehaviourOf{x_0})\in G(\TheBehaviourOf{-})$ (this holds by definition of $\TheBehaviourOf{-}$);
\item that $\TheElement$ and $\TheBehaviourOf{\TheElement}$ have the same observations; i.e., that $\TheOutputOf{\TheElement}\theta_\Gamma(\TheBehaviourOf{\TheElement})$; and
\item that the next states of $\TheElement$ are in relation with the next states of $\TheBehaviourOf{\TheElement}$; i.e., that for all $i\in \TheSetOfInputs$, $(\TheElement^\TheInput, \delta_\Gamma(\TheBehaviourOf{\TheElement},\TheInput))\in G(\TheBehaviourOf{-})$.
\end{enumerate}
For 2.: By the definitions of $\TheBehaviourOf{-}$ and of $\theta_\Gamma$, we see that
\begin{align*}
\theta_\Gamma(\TheBehaviourOf{\TheElement})=\TheBehaviourOf{\TheElement}(\varepsilon)=\TheOutputOf{\TheElement},
\end{align*}
so this condition is satisfied.
\newline
For 3.: by the definitions of $\TheBehaviourOf{-}$ and of $\delta_\Gamma$, we have, for $i\in \TheSetOfInputs$, that
\begin{align*}
\delta_\Gamma(\TheBehaviourOf{\TheElement},\TheInput)(\TheInputSequence)=\TheBehaviourOf{\TheElement}(\TheInput:\TheInputSequence)= \TheBehaviourOf{\TheElement^\TheInput}(\TheInputSequence)%=\delta_\Gamma(\theta(x):\TheBehaviourOf{\delta(x)})=\TheBehaviourOf{\delta(x)}.
\end{align*}
Since $(\TheElement^\TheInput,\TheBehaviourOf{\TheElement^\TheInput})\in G(\TheBehaviourOf{-})$ for all $\TheInput \in \TheSetOfInputs$, we conclude that $(\TheElement^\TheInput, \delta_\Gamma(\TheBehaviourOf{\TheElement},\TheInput))\in G(\TheBehaviourOf{-})$, so this condition is also satisfied.

As conditions 1., 2., and 3. are satisfied, we infer that $G(\TheBehaviourOf{-})$ is a bisimulation between $\TheCoalgebra$ and $\Gamma({\TheBehaviourOf{x_0}})$. Since coalgebraic logics respect bisimulations, we conclude, for every formula $\TheFormula$, that $(\TheCoalgebra, x_0)\vDash \TheFormula$ if and only if $(\Gamma(\TheBehaviourOf{x_0}), \TheBehaviourOf{x_0})\vDash \TheFormula$.\qed
\end{proof}

\bibliographystyle{lncs/splncs03}

\end{document}